\documentclass[amsthm]{JSC_LaTex_2007_Mar_12/elsart}

\usepackage{JSC_LaTex_2007_Mar_12/yjsco}
\usepackage{JSC_LaTex_2007_Mar_12/natbib}
\usepackage{myMacLite}

\newcommand{\balg}{\begin{algorithm}}
\newcommand{\ealg}{\end{algorithm}}

\renewcommand\refname{References and Notes}
\renewcommand{\abstractname}{Abstract}
\renewcommand{\thefootnote}{\fnsymbol{footnote}}

\def\qed{\hfil {\vrule height5pt width2pt depth2pt}}

\def\bref#1{(\ref{#1})}

\def\R{{\mathbb{R}}}
\def\C{{\mathbb{C}}}
\def\Q{{\mathbb{Q}}}
\def\I{{\mathbb{I}}}
\def\B{{\bf B}}

\def\PS{ {\mathcal{P}} }
\def\GB{{\mathcal{G}}}
\def\IS{{\mathcal{I}}}

\def\RB{\hbox{\rm{RB}}}
\def\Dis{\hbox{\rm{Dis}}}

\begin{document}
\begin{frontmatter}

%

%
%


\title{Root Isolation of Zero-dimensional\\ Polynomial Systems with \\Linear Univariate Representation$^1$}
\author{Jin-San Cheng, Xiao-Shan Gao, Leilei Guo}

\address{
KLMM, Institute of Systems Science, AMSS, Chinese Academy of Sciences\\
Email: xgao@mmrc.iss.ac.cn,jcheng@amss.ac.cn}


\footnotetext[1]{Partially supported by a National Key Basic
Research Project of China and by a grant from NSFC.}

\begin{abstract}
In this paper, a linear univariate representation for the roots of a
zero-dimensional polynomial equation system is presented, where the
roots of the equation system are represented as linear combinations
of roots of several univariate polynomial equations. The main
advantage of this representation is that the precision of the roots
can be easily controlled. In fact, based on the linear univariate
representation, we can give the exact precisions needed for
isolating the roots of the univariate equations in order to obtain
the roots of the equation system to a given precision. As a
consequence, a root isolation algorithm for a zero-dimensional
polynomial equation system can be easily derived from its linear
univariate representation.
\end{abstract}

\end{frontmatter}

\section{Introduction}
Solving polynomial equation systems is a basic problem in the field
of computational science and has important engineering applications.
In most cases, we consider zero-dimensional polynomial systems. We
will discuss how to solve this kind of systems in this paper. In
particular, we will consider how to isolate the complex roots for
such a system.

One of the basic methods to solve polynomial equation systems is
based on the concept of separating elements, which can be traced
back to Kronecker \cite{kronecker} and has been studied extensively
in the past twenty years
\cite{abrw,canny,chengjsc,gaochou,gm,gh,gls,kmh,kff,ll,renegar,rur,ynt}.
The idea of the method is to introduce a new variable $t=\sum_i
c_ix_i$ which is a linear combination of the variables to be solved
such that $t=\sum_i c_ix_i$ takes different values when evaluated
at different roots of the polynomial equation system
$\mathcal{P}=0$. In such a case, we say that $t$ is a {\bf
separating element} for $\mathcal{P}=0$. If $t=\sum_i c_ix_i$ is a
{separating element} for $\mathcal{P}=0$,  the roots of
$\mathcal{P}=0$ have the following rational univariate
representation (RUR):
 $$f(t)=0, x_i=R_i(t), i=1,...,n,$$
where $f\in \Q[t]$ and $R_i(t)$ are rational functions in $t$. As a
consequence, solving multi-variate equation systems is reduced to
solving a univariate equation $f(t)=0$ and to substituting the roots
of $f(t)=0$ into rational functions $R_i(t)$.
Along this line, better complexity bounds and effective software
packages for solving polynomial equations such as the Maple package
RootFinding by Rouillier \cite{rur} and the Magma package Kronecker
by Giusti, Lecerf, and Salvy \cite{gls} are given.

The above approaches still have the following problem: for an
isolation interval $[a,b]$ of a real root $\alpha$ of $f(u)=0$, to
determine the isolation interval of $x_i=R_i(\alpha)$ under a given
precision is not a trivial task.
In this paper, we propose a new representation for the roots of a
polynomial system which will remedy this drawback.

\begin{figure}[ht]
\centering
\begin{minipage}{0.9\textwidth}
\centering
 \includegraphics[scale=0.32]{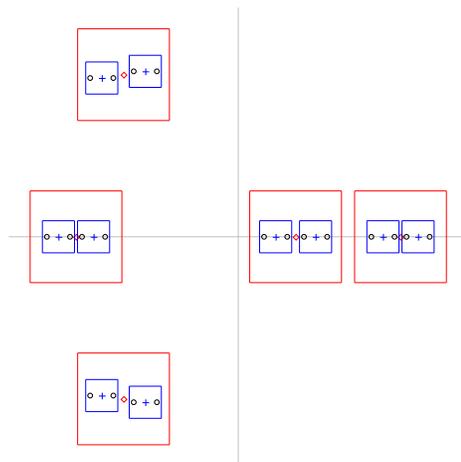}
 \caption{The distribution of the roots of $T_i(x)=0 (i=1,2,3)$. The red diamonds (blue crosses, black circles) are roots of $T_1(x)=0$ ($T_2(x)=0,\,T_3(x)=0$) and red (blue) boxes are neighborhoods for the red diamonds (blue crosses).}
\label{fig-slope}
\end{minipage}
\end{figure}

%
In the ISSAC paper \cite{lgp-bi}, based on ideas similar to
separating elements, a local generic position method is introduced
to solve bivariate polynomial systems and experimental results show
that the method is quite efficient for solving equation systems with
multiple roots.
In this paper, we extend the method to solve general
zero-dimensional polynomial systems. A local generic position for a
polynomial equation system $\PS=0$ is also a linear combination of
the variables to be solved: $t=\sum_i c_ix_i$ which satisfies two
conditions.
First,  $t_k=\sum_{i=1}^k c_ix_i$ is a separating element of
$\mathcal{P}_k =(\mathcal{P})\cap\Q[x_1,\ldots,x_k]$ for
$k=2,\ldots,n$, and the roots of $\mathcal{P}_k=0$ have a one-to-one
correspondence with the roots of a univariate equation $T_k(t_k)=0$.
Second, for a root $\xi=(\xi_1,\ldots,\xi_k)$ of $\mathcal{P}_k=0$
represented by a root $\eta$ of $T_k(t_k)=0$, all the roots $\eta_j$
of $T_{k+1}(t_{k+1})=0$ corresponding to the roots of
$\mathcal{P}_{k+1}=0$, say $\xi_j=(\xi,\xi_{k+1,j})$, ``lifted" from
$\xi$ are projected into a fixed square neighborhood of $\eta$. This
``local" property is illustrated in Figure \ref{fig-slope}.
We prove that if $t_n=\sum_{i=1}^n c_ix_i$ is a local generic
position for $\mathcal{P}$, then the roots of $\mathcal{P}=0$ can be
be represented as special linear combinations of the roots of
univariate equations $T_k(t_k)=0,k=1,\ldots,n$:
 $$\{ (\alpha_1,\frac{\alpha_2-\alpha_1}{s_1},\ldots,\frac{\alpha_n-\alpha_{n-1}}{s_1\cdots s_{n-1}})\,|\, T_k(\alpha_k)=0 \},$$
where $s_j$ are certain positive rational numbers and the
$\alpha_{j+1}$ matching $\alpha_j$ are in certain square
neighborhood of $\alpha_j$ to be defined in Section 2.
Such a representation is called a {\bf linear univariate
representation} (LUR for short)  of the polynomial system.

The main advantage of the LUR representation is that the precision
of the roots can be easily controlled. For RUR, computing solutions
with a given precision is not a trivial task as we mentioned before.
It is not easy to know with which precision to isolate the roots of
$f(t)=0$ is enough in order for the roots of the system $x_i=R_i(t)$
to satisfy a given precision.
%
%
For LUR, precision control becomes very easy. We can give an
explicit formula for the precision of the roots of $T_i(x)=0$ in
order to obtain the roots of the system with a given precision. So
we can obtain the solutions of the system by refining the roots of
$T_i(x)=0$ at most once. Another advantage of LUR is that when we
isolate the roots of $T_{i+1}(x)=0$, we need only to consider a
fixed neighborhood of each root of $T_i(x)=0$.

We propose an algorithm to compute an LUR for a zero-dimensional
polynomial system. The key ingredients of the algorithm are to
estimate the root bounds of $\PS=0$ and to estimate the separating
bounds for the roots of $\PS_{k+1}=0$ lifted from a root of
$\PS_{k}=0$. We adopt a computational approach to estimate such
bounds in order to obtain tight bound values. For the root bounds of
$\PS=0$, we use Gr\"obner basis computation to obtain the generating
polynomial of the principal ideal $(\PS)\cap\Q[x_i]$ and use this
polynomial to estimate the root bound for the $x_i$ coordinates of
the roots of $\PS=0$.
The separating bounds for $\PS_k=0$ are obtained from the isolating
boxes for the roots of the $T_k(x)=0$. These bounds in turn will be
used to compute the isolating boxes for the roots of $\PS_{k+1}=0$.
Hence, the algorithm to compute an LUR also gives a set of isolating
boxes for the roots of $\PS=0$.

In Section 2, we give the definition of LUR and the main result of
the paper.
In Section 3, we present an algorithm to compute an LUR of a
zero-dimensional polynomial system as well as a set of isolation
boxes of the roots of the equation system.
In Section 4, we provide some illustrative examples.
We conclude the paper in Section 5.

\section{Linear univariate representation}
In this section, we will define LUR and prove its main properties.
Let
$$\mathcal{P} = \{f_1(x_1,\ldots,x_n),\ldots, f_s(x_1,\ldots,x_n) \}$$
be a zero-dimensional polynomial system in $\Q[x_1,\ldots,x_n]$,
where $\Q$ is the field of rational numbers. Let
 $$\IS_i = (\mathcal{P})\cap\Q[x_1,\ldots,x_i],i=1,\ldots,n,$$
where $(\PS)$ is the ideal generated by $\PS$.
We use $V_{\C}(\mathcal{P})$ to denote its complex roots in $\C^n$.

Since we will use rectangles to isolate complex numbers, we adopt
the following norm for a complex number $c = x + yi$:
\begin{eqnarray}\label{eq-dis}|c | = \max\{|x|,|y|\}.\end{eqnarray}
The ``distance\footnote{The results in this section are also valid
if we use the usual distance for complex numbers.}" between two
complex numbers $c_1$ and $c_2$ is defined to be $|c_1-c_2|$. It is
easy to check that this is indeed a distance satisfying the
inequality $|c_1-c_2|\le |c_1-c_3| + |c_3-c_2|$ for any complex
number $c_3$.
Let $c_0$ be a complex number and $r$ a positive rational number.
Then the set of points having distance less than $r$ with $c_0$,
that is $\{ c_1\in\C \, |\, |c_1-c_0|< r \}$, is an open square with
$c_0$ as the center.

By an LUR, we mean a set like
\begin{eqnarray}\label{eq-lur}
 \{T_1(x),\ldots,T_n(x), s_i,d_i, i=1,\ldots,n-1\},
\end{eqnarray}
where $T_i(x)\in \Q[x]$ are univariate polynomials, $s_i$ and $d_i$
are positive rational numbers. The {\bf roots} of \bref{eq-lur} are
defined to be
\begin{eqnarray}\label{eq-rt}
&&
\{(\alpha_1,\frac{\alpha_2-\alpha_1}{s_1},\ldots,\frac{\alpha_n-\alpha_{n-1}}{s_1\cdots\,s_{n-1}})
 \, |\, T_i(\alpha_i)=0, i=1,\ldots,n \hbox{ and }\nonumber\\
&& |\alpha_{i+1}-\alpha_i|<s_1\cdots s_{i-1}d_i, i=1,\ldots,n-1 \}
\end{eqnarray}
where $s_0=1$. Geometrically, we match a root $\alpha_i$ of
$T_{i}(x)=0$ with those roots of $T_{i+1}(x)=0$ inside a squared
neighborhood centered at $\alpha_i$. See Figure \ref{fig-slope} for
an illustration.
An {\bf LUR for $\mathcal{P}$} is a set of form \bref{eq-lur} whose
roots are exactly the roots of $\mathcal{P}=0$.

It is clear that an LUR represents the roots of $\mathcal{P}$ as
linear combinations of the roots of some univariate polynomial
equations. The LUR representation has the following advantage: we
can easily derive the precision of the roots of $\mathcal{P}=0$ from
that of the univariate equations as shown by the following lemma.

\begin{lem}
Let \bref{eq-lur} be an LUR for a polynomial system $\PS=0$. If
$\alpha_i$ is a root of $T_i(x)=0 (1\le i\le n)$ and
$\overline{\alpha}_i$ is an approximation of $\alpha_i$ with
precision $\epsilon_i$, then the approximate root
$(\overline{\alpha}_1,\frac{\overline{\alpha}_2-\overline{\alpha}_1}{s_1},\ldots,$
$\frac{\overline{\alpha}_n-\overline{\alpha}_{n-1}}{s_1\cdots\,s_{n-1}})$
of $\mathcal{P}=0$ has precision $\max\{\epsilon_1,
\frac{\epsilon_2+\epsilon_1}{s_1},\ldots,\frac{\epsilon_n+\epsilon_{n-1}}{s_1\cdots\,s_{n-1}}\}$.
\end{lem}
\begin{pf} Since $x_i
=\frac{\alpha_i-\alpha_{i-1}}{s_1\cdots\,s_{i-1}}$ and the
approximate root $\overline{\alpha}_i$ of $\alpha_i$ has precision
$\epsilon_i$,  the approximate root $\overline{x}_i
=\frac{\overline{\alpha}_i-\overline{\alpha}_{i-1}}{s_1\cdots\,s_{i-1}}$
has precision no larger than
$\frac{\epsilon_i+\epsilon_{i-1}}{s_1\cdots\,s_{i-1}}$.
\end{pf}

For a zero-dimensional polynomial system $\mathcal{P}$, let $d_i,
r_{i}$ ($i=1,\ldots,n$), and $s_i$ ($i=1,\ldots,n-1$) be positive
rational numbers satisfying
\begin{eqnarray}
 D_i&=&\,\min\{\frac{1}{2}|\alpha-\beta|, \forall \eta\in V_{\C}(\IS_{i-1}),
   (\eta,\alpha), (\eta,\beta)\in V_{\C}(\IS_i), \alpha\ne\beta\},\label{eq-dn0}\\
 d_i&<&\,\min\{D_i, \frac{d_{i-1}}{2s_{i-1}}\},\label{eq-dn1}\\
 r_{i}&>&2\max\{|\alpha_{i}|, \forall (\alpha_1,\ldots,\alpha_{i})\in V_{\C}(\IS_{i}) \}, \label{eq-dn2}\\
 s_i&\le&\frac{d_i}{r_{i+1}}\label{eq-dn3}
\end{eqnarray}
where $s_0=1, d_0 = +\infty$.
Geometrically, $D_i$ is half of the root separation bound for roots
of $\IS_i$ considered as points on a ``fiber" over each root of
$\IS_{i-1}$, $r_i$ is twice the root bound for the $i$-th
coordinates of the roots of $\IS_i$, and $s_i$, the inverse of the
slope of certain line, is a key parameter to be used in our method.
If $\forall \eta\in V_{\C}(\IS_{i-1})$, $\#\{\alpha|(\eta,\alpha)\in
V_{\C}(\IS_i)\}=1$, we can choose any positive number as $d_i$.

For $s_i$ satisfying \bref{eq-dn3}, consider the ideal
\begin{equation}\label{eq-bi}
 \bar{\IS_i} = (\IS_i\cup\{x - x_1 - s_1 x_2 - \cdots -
 s_1\cdots s_{i-1} x_i\}),\end{equation}
 where $x$ is a new variable.
%
It is clear that $\bar{\IS_i}$ is a zero-dimensional ideal in
$\Q[x_1,\ldots,x_i,x]$. And the elimination ideal
$(\bar{\IS_i})\cap\Q[x]$ is principal. Let $T_i(x)$ be the generator of this
ideal:
\begin{eqnarray}\label{eq-tii}
(\bar{\IS_i})\cap\Q[x] = (T_i(x)).
\end{eqnarray}
The following is the main result of this paper.
\begin{thm}\label{th-1}
If $d_i,  s_i$ satisfy conditions \bref{eq-dn1}, \bref{eq-dn3} and
$T_i$ is defined in \bref{eq-tii}, then the corresponding set
\bref{eq-lur} is an LUR for $\mathcal{P}$.
\end{thm}

We will prove two lemmas which will lead to a proof for the theorem.
For a root $\eta_i$ of $T_i(x)=0$, let
\begin{eqnarray}\label{eq-nbh}
\mathbb{S}_{\eta_{i}}&=& \{\eta\in\C\, | \, |\eta-\eta_i|<\rho_i \},
i=1,\ldots,n
\end{eqnarray}
where $\rho_i=s_1\cdots s_{i-1}d_{i}, (s_0=1)$.
Note that $\mathbb{S}_{\eta_{i}}$ is an open square whose center is
$\eta_i$ and whose edge has length $2\rho_i$.
With this notation, the {roots} of \bref{eq-lur} can be written as
\begin{eqnarray}\label{eq-rt1}
&&
\{(\alpha_1,\frac{\alpha_2-\alpha_1}{s_1},\ldots,\frac{\alpha_n-\alpha_{n-1}}{s_1\cdots\,s_{n-1}})
 \, |\, T_i(\alpha_i)=0, i=1,\ldots,n \hbox{ and }\nonumber\\
&& \alpha_{i+1}\in\mathbb{S}_{\alpha_{i}}, i=1,\ldots,n-1 \}
\end{eqnarray}
In Figure \ref{fig-slope}, $\mathbb{S}_{\eta_{i}}$ are interior
parts of the squares.
We have
\begin{lem}\label{lm-ww0}
Under assumptions of Theorem \ref{th-1}, we have
$\mathbb{S}_{\eta_{i+1}} \subset \mathbb{S}_{\eta_{i}},
$i=1,\ldots,n-1$,$ where $(\xi_1,\ldots,$ $\xi_{i+1})\in
V_{\C}(\IS_{i+1})$ and
\begin{eqnarray}
\eta_{i}&=&\xi_1+s_1\xi_2+\cdots+s_1\cdots
s_{i-1}\,\xi_{i},\label{eq-ww4}\\
\eta_{i+1}&=&\xi_1+s_1\xi_2+\cdots+s_1\cdots
s_{i-1}\,\xi_{i}+s_1\cdots
 s_{i}\,\xi_{i+1} = \eta_{i} +s_1\cdots
 s_{i}\,\xi_{i+1}.\label{eq-ww3}
\end{eqnarray}
\end{lem}
\begin{proof}
From the definition of $\bar{\IS_i}$ in \bref{eq-bi}, $\eta_{i}$ is
a root of $T_i(x)=0$, $\eta_{i+1}$ is a root of $T_{i+1}(x)=0$, and
each root of $T_{i+1}(x)=0$ has the form \bref{eq-ww3}.

We first prove that  $\eta_{i+1}\in \mathbb{S}_{\eta_{i}}$.
Using \bref{eq-dn2} and \bref{eq-dn3}, we have
\begin{equation}\label{eq-ww1}
 |\eta_{i+1}-\eta_{i}|=s_1\cdots s_{i}|\xi_{i+1}|
 <\frac{1}{2}s_1\cdots s_{i}r_{i+1} \le \frac{1}{2}s_1\cdots
s_{i-1}d_{i} = \frac{1}{2}\rho_i.\end{equation}
As a consequence, $\eta_{i+1}$ is in $\mathbb{S}_{\eta_{i}}$.

We now prove that $\mathbb{S}_{\eta_{i+1}}\subset
\mathbb{S}_{\eta_{i}}$. By \bref{eq-dn1}, we have $\rho_{i+1} =
s_1\cdots s_{i}d_{i+1} < \frac{1}{2}  s_1\cdots s_{i-1}d_{i} =
\frac{1}{2}\rho_i$. Therefore, for any $\eta\in
\mathbb{S}_{\eta_{i+1}}$, by \bref{eq-ww1}, we have $|\eta-\eta_i|
\le |\eta-\eta_{i+1}| + |\eta_{i+1}-\eta_i| < \rho_{i+1} +
\frac{1}{2}\rho_i < \rho_i$. Hence $\eta\in \mathbb{S}_{\eta_{i}}$
and the lemma is proved.
\end{proof}

For rational numbers $a_j$, we call $f_i=\Sigma_{j=1}^i a_j x_j$ a
{\bf separating element} of $\IS_i$, if $\forall \alpha,\beta\in
V_{\mathbb{C}}(\IS_i)$,  $\alpha\ne\beta$ implies $f_i(\alpha)\neq
f_i(\beta)$ (see paper \cite{rur}).

Theorem \ref{th-1} follows from (d) of the following lemma.
\begin{lem}\label{lm-ww1}
Under assumptions of Theorem \ref{th-1}, for $i=1,\ldots,n$, we have
\begin{description}
\item[] (a) $x=x_1+s_1\,x_2+\cdots+s_1\cdots s_{i-1}x_i$ is a separating element of
$\IS_i$.

\item[] (b) Each root $\eta_{i}$ of $T_{i}(x)=0$ is in an $\mathbb{S}_{\eta_{i-1}}$
for a root $\eta_{i-1}$ of $T_{i-1}(x)=0$. Furthermore, if
$\eta_{i-1}=\xi_1+s_1\,\xi_2+\cdots+s_1\cdots s_{i-2}\xi_{i-1}$,
then all roots of $T_{i}(x)=0$ in $\mathbb{S}_{\eta_{i-1}}$ are of
the following form
\begin{equation}\label{eq-ww2} \eta_{i}=\eta_{i-1} +s_1\cdots
 s_{i-1}\,\xi_{i}\end{equation}
where $(\xi_1,\ldots,\xi_{i-1},\xi_{i})\in V_{\C}(\IS_{i})$.

\item[] (c) $\mathbb{S}_{\eta_{i}}$ are disjoint for all roots $\eta_i$ of
$T_i(x)=0$.

\item[] (d) $(T_1(x),\ldots,T_i(x), s_j, d_j, j=1,\dots,i-1)$ is an LUR for $\IS_i$.
\end{description}
\end{lem}
\begin{proof}
We will prove the lemma by induction on $k=i$.
For $k=1$, since $(\IS_1) = (T_1(x))$,  statements (a) and (d) are
obviously true. We do not need prove (b).
From \bref{eq-dn1}, we have $d_1<\,\min\{\frac{1}{2}|\alpha-\beta|,
\forall \alpha, \beta\in V_{\C}(\IS_1)= V_{\C}(T_1),
\alpha\ne\beta\}$. As a consequence, $\mathbb{S}_{\eta_{1}}$ are
disjoint for all roots $\eta_1$ of $T_1(x)=0$. Statement (c) is
proved.

Suppose that the result is correct for $k=1,\ldots,i$. We will prove
the result for $k=i+1$.

We first prove statement (a). Let $\xi=(\xi_1,\ldots,\xi_{i+1})$ and
$\beta=(\beta_1,\ldots,\beta_{i+1})$ be two distinct elements in
$V_{\C}(\IS_{i+1})$.
We consider two cases. If $(\xi_1,\ldots,\xi_{i})$ is different from
$(\beta_1,\ldots,\beta_{i})$, then by the induction hypothesis
$\eta_{i}=\xi_1+ s_1\xi_2+\cdots+s_1\cdots s_{i-1} \xi_{i}$ is also
different from $\theta_{i}=\beta_1+ s_1\beta_2+\cdots+ s_1\cdots
s_{i-1} \beta_{i}$. By (c) of the induction hypothesis,
$\mathbb{S}_{\eta_{i}}$ and $\mathbb{S}_{\theta_{i}}$ are disjoint.
By Lemma \ref{lm-ww0},  $\eta_{i+1}=\eta_{i}+s_1\cdots s_{i}
\xi_{i+1}\in \mathbb{S}_{\eta_{i}}$ and
$\theta_{i+1}=\theta_{i}+s_1\cdots s_{i} \beta_{i+1}\in
\mathbb{S}_{\theta_{i}}$. Then, in this case we have
$\eta_{i+1}\ne\theta_{i+1}$.
In the second case, we have
$(\xi_1,\ldots,\xi_{i})=(\beta_1,\ldots,\beta_{i})$. Then,
$\eta_{i}=\theta_{i}$ and $\xi_{i+1}\ne \beta_{i+1}$. It is clear
that $\eta_{i+1} = \eta_{i} +s_1\cdots s_{i} \xi_{i+1}$ is
 different from $\theta_{i+1} = \theta_{i} +s_1\cdots s_{i}
\beta_{i+1}$. Thus, (a) is proved.

We now prove statement (b). Use notations in \bref{eq-ww4} and
\bref{eq-ww3}. By Lemma \ref{lm-ww0}, we have
$\eta_{i+1}\in\mathbb{S}_{\eta_{i}}$. Then, each root of
$T_{i+1}(x)=0$ is in an $\mathbb{S}_{\eta_{i}}$ for a root
$\eta_{i}$ of $T_{i}(x)=0$.
Let $(\beta_1,\ldots,\beta_{i+1})\in V_{\C}(\IS_{i+1})$ such that
$\theta_{i+1}=\beta_1+ s_1\beta_2+\cdots+ s_1\cdots s_{i}
\beta_{i+1}$ is another element in  $\mathbb{S}_{\eta_{i}}$.
We claim that $(\beta_1,\ldots,\beta_{i})$ must be the same as
$(\xi_1,\ldots,\xi_{i})$. Otherwise, by the induction hypothesis
(a), $\theta_{i}=\beta_1+ s_1\beta_2+\cdots+ s_1\cdots s_{i-1}
\beta_{i}$ is different from $\eta_i$. By the induction hypothesis
(c), $\mathbb{S}_{\eta_{i}}$ and  $\mathbb{S}_{\theta_{i}}$ are
disjoint which is impossible since by Lemma \ref{lm-ww0},
$\theta_{i+1}\in\mathbb{S}_{\eta_{i}}$ and
$\theta_{i+1}\in\mathbb{S}_{\theta_{i}}$.
Thus, $(\beta_1,\ldots,\beta_{i})=(\xi_1,\ldots,\xi_{i})$ and hence
$\theta_{i+1}=\eta_i + s_1\cdots s_{i} \beta_{i+1}$. This proves
equation \bref{eq-ww2} and hence statement (b).

We now prove statement (c).
Use notations in \bref{eq-ww4} and \bref{eq-ww3}. By Lemma
\ref{lm-ww0}, $\mathbb{S}_{\eta_{i+1}}\subset
\mathbb{S}_{\eta_{i}}$.
As a consequence, we need only to prove that the squares
$\mathbb{S}_{\eta_{i+1}}$ contained in the same
$\mathbb{S}_{\eta_{i}}$ are disjoint. Let $\eta_{i+1}, \theta_{i+1}$
be two roots of $T_{i+1}(x)=0$ in $\mathbb{S}_{\eta_{i}}$. By
statement (b) just proved, we have
$$ \eta_{i+1}= \eta_i+s_1\cdots s_{i}\xi_{i+1},
   \theta_{i+1}=\eta_i +s_1\cdots s_{i}\beta_{i+1}$$
where $\eta_i$ is defined in \bref{eq-ww4} and
$(\xi_1,\ldots,\xi_i,\xi_{i+1})$, $(\xi_1,\ldots,\xi_i,\beta_{i+1})$
are roots of $\IS_{i+1}$.   Then, by \bref{eq-dn1},
$$|\eta_{i+1}-\theta_{i+1}|=s_1\cdots s_{i}|\xi_{i+1}-\beta_{i+1}|>2\,s_1\cdots s_{i}\,d_{i+1} = 2\rho_{i+1}.$$
So, $\mathbb{S}_{\eta_{i+1}}$ and $\mathbb{S}_{\theta_{i+1}}$ are
disjoint. Statement (c) is proved.

Finally, we prove statement (d). Let
$\xi=(\xi_1,\ldots,\xi_{i+1})\in V_{\C}(\IS_{i+1})$ and
$\eta_{j}=\xi_1+s_1\xi_2+\cdots+s_1\cdots
 s_{j-1}\,\xi_{j},j=1,\ldots,i+1$. By the induction hypothesis, we
have
$(\xi_1,\ldots\xi_{i})=(\eta_1,\frac{\eta_2-\eta_1}{s_1},\ldots,\frac{\eta_i-\eta_{i-1}}{s_1\cdots\,s_{i-1}})$
where $|\eta_{j+1}-\eta_j|<s_1\cdots s_{j-1}d_j, j=1,\ldots,i$. Note
that the inequality is equivalent to that $\eta_{j+1}\in
\mathbb{S}_{\eta_{j}}$.
By \bref{eq-ww2}, we can recover the $\xi_{i+1}$ with the following
equation
 $$ \xi_{i+1} = \frac{\eta_{i+1}-\eta_{i}}{s_1\cdots s_{i}}.$$
From Lemma \ref{lm-ww0}, we have
$\eta_{i+1}\in\mathbb{S}_{\eta_{i}}$ or equivalently
$|\eta_{i+1}-\eta_{i}| < s_1\cdots s_{i-1}d_{i}$. Then the root
$(\xi_1,\ldots\xi_{i+1})=(\eta_1,\frac{\eta_2-\eta_1}{s_1},\ldots,\frac{\eta_{i+1}-\eta_{i}}{s_1\cdots\,s_{i}})$
is a root of the LUR $(T_1(x),\ldots,T_{i+1}(x), s_j, d_j,
j=1,\dots,i)$. We thus proved that the roots of $\IS_{i+1}$ are the
same as the roots of the LUR and hence statement (d).
\end{proof}

We have the following corollaries.
\begin{cor}\label{cor-l1}
If \bref{eq-lur} is an LUR for a polynomial system $\PS$, then the
roots of $\IS_i=0$ are in a one to one correspondence with the roots
of $T_i(x)=0$ for $i=1,\ldots,n$.
\end{cor}
\begin{proof}
Let $\xi=(\xi_1,\ldots,\xi_{i})\in V_{\C}(\IS_{i})$. Then
$\eta_{i}=\xi_1+s_1\xi_2+\cdots+s_1\cdots
 s_{i-1}\,\xi_{i}$ is a root of $T_i(x)=0$.
By (a) of Lemma \ref{lm-ww1}, this mapping is injective. This
mapping is clearly surjective.
\end{proof}

\begin{cor}
The real roots of $\PS=0$ are in  a one to one correspondence with
the real roots of $T_n(x)=0$. More precisely, if $\alpha_n$ is a
real root of $T_n(x)=0$, then in the corresponding root
$(\alpha_1,\frac{\alpha_2-\alpha_1}{s_1},\ldots,\frac{\alpha_n-\alpha_{n-1}}{s_1\cdots\,s_{n-1}})$
of $\PS=0$, $\alpha_i$ is a real root of $T_i(x)=0,i=1,\ldots,n-1$.
\end{cor}
\begin{proof}
For each root $\eta$ of $T_{i-1}(x)=0$, let  $\mathbb{S}_{\eta}$ be
the open square neighborhood of $\eta$ defined in \bref{eq-nbh}.
We claim that a real root of $T_i(x)=0$ cannot be in
$\mathbb{S}_{\eta}$ for a complex root $\eta$ of $T_{i-1}(x)=0$.
Since $T_{i-1}(x)$ has rational numbers as coefficients, the complex
roots of $T_{i-1}(x)=0$ appear as pairs which are symmetric with the
real axis and the open square neighborhoods for a pair of complex
roots are disjoint. Then the open square neighborhood of any complex
root has no intersection with the real axis. This proves the claim.
As a consequence, if $\alpha_n$ is a real root of $T_n(x)=0$, then
$\alpha_n$ is in the open square neighborhood of a real root
$\alpha_{n-1}$ of $T_{n-1}(x)=0$. Repeating the process, we obtain a
real root
$(\alpha_1,\frac{\alpha_2-\alpha_1}{s_1},\ldots,\frac{\alpha_n-\alpha_{n-1}}{s_1\cdots\,s_{n-1}})$
for $\PS=0$ where all $\alpha_i$ are real numbers.
The other side is obvious: a real root of $\PS=0$ will correspond to
a real root of $T_n(x)=0$.
\end{proof}
From the lemma, we can consider the real roots of an LUR if we only
interest in the real roots of $\PS=0$.

\section{Algorithm for computing an LUR and roots isolation}
In this section, we will present an algorithm to compute an LUR for
a zero-dimensional polynomial system. The algorithm will isolate the
roots of the system in $\mathbb{C}^n$ at the same time.

\subsection{Complex isolation intervals and isolation boxes}
We introduce some basic concepts of interval computation. For more
details, we refer to \cite{intervalbook}.

Let $\intbox \Q$ denote the set of intervals of the form $[a,b]$,
where $a\le b\in \Q$. The {\bf length} of an interval $I=
[a,b]\in\intbox \Q$ is defined to be $|I| = b-a$.
%
%
Assuming $a_1\le a_2$, we define the distance between two intervals
as
$$\Dis([a_1,b_1],[a_2,b_2])=
\begin{cases}
a_2-b_1, \,\hbox{ if } [a_1,b_1]\cap[a_2,b_2]=\emptyset,\\
0, \hspace{1cm}\hbox{        otherwise}.
\end{cases}
$$

A pair of intervals  $\langle I,J\rangle$ is called a {\bf complex
interval}, which represents a rectangle in the complex plane. A
complex number $\langle \alpha, \beta \rangle=\alpha+\beta
\mathfrak{i}$ ($\mathfrak{i}^2=-1$) is said to be in a complex
interval $\langle I, J \rangle$ if $\alpha\in I$ and  $\beta\in J$.
 The length of a complex interval $\langle I,J\rangle$ is defined to
be $|\langle I,J\rangle|=\max\{|I|, |J|\}$.
We define the distance between two complex intervals as {\small
\begin{equation}\label{eq-rtsep} \Dis(\langle [a_1,b_1],
[p_1,q_1]\rangle,\langle [a_2,b_2], [p_2,q_2]\rangle)=
\max\{\Dis([a_1,b_1],[a_2,b_2]),\Dis([p_1,q_1],[p_2,q_2]\}.
\end{equation}}

A set $\mathcal{S}$ of disjoint complex intervals  is called {\bf
isolation intervals} of $T(x)=0$ if each interval in $\mathcal{S}$
contains only one root of $T(x)=0$ and each root of $T(x)=0$ is
contained in one interval in $\mathcal{S}$.
Methods to isolate the complex roots of a univariate polynomial
equation are given in \cite{collins-c1,pinkert,sy,wilf}.

Let $\intbox \C$ denote the set of complex intervals. An element
$\langle I^{\R}_1, I^{\I}_1\rangle\times\cdots\times\langle
I^{\R}_n, I^{\I}_n\rangle$ in $\intbox \C^n$ is called a {\bf
complex box}.
A set $\mathcal{S}$ of {\bf isolation boxes} for a zero dimensional
polynomial system $\mathcal{P}$ in $\Q[x_1,\ldots,x_n]$ is a set of
disjoint complex boxes in  $\intbox \C^n$ such that each box in
$\mathcal{S}$ contains only one root of $\mathcal{P}=0$ and each
root of $\mathcal{P}=0$ is in one of the boxes.
Furthermore,  if each box  $\B=\langle I^{\R}_1,
I^{\I}_1\rangle\times\cdots\times\langle I^{\R}_n, I^{\I}_n\rangle$
in $\mathcal{S}$ satisfies $\max\limits_i\{|I^{\R}_i|,
|I^{\I}_i|\}\le\epsilon$, then $\mathcal{S}$ is called an {\bf
$\epsilon$-isolation boxes} of $\mathcal{P}=0$.
The aim of this paper is to compute a set of $\epsilon$-isolation
boxes for a zero-dimensional polynomial system $\mathcal{P}$.

\subsection{Gr\"obner basis and computation of $r_i$ and $T_i(x)$}
In this subsection, we will show how to use Gr\"obner basis to
compute $r_i$ defined in \bref{eq-dn2} and $T_i(x)$ defined in
\bref{eq-bi} supposing the parameters $s_i$ are given.

Let $\PS\subset\Q[x_1,\ldots,x_n]$ be a zero-dimensional polynomial
system. Then $\mathcal{A}=\mathbb{Q}[x_1,\ldots,x_n]/$ $(\PS)$ is a
finite dimensional linear space over $\Q$. Let $\GB$ be a
Gr$\ddot{o}$bner basis of $\PS$ with any ordering. Then the set of
remainder monomials
$$\B=\{x_1^{t_1}\cdots x_n^{t_n}|x_1^{t_1}\cdots x_n^{t_n} \hbox{ is not divisible by the leading term of any element of } \GB\}$$
forms a basis of $\mathcal{A}$ as a linear space over $\Q$, where
$t_i$ are non-negative integers.

Let $f\in\Q[x_1,\ldots,x_n]$. Then $f$ gives a multiplication map
$$M_f: \mathcal{A}\longrightarrow \mathcal{A}$$
defined by $M_f(p)=f p$ for $p\in \mathcal{A}$.
It is clear that $M_f$ is a linear map. We can construct the matrix
representation for $M_f$ from  $\B$ and $\GB$. The following theorem
is a basic property for $M_f$ \cite{lazard1}.

\begin{thm}[Stickelberger's Theorem]\label{thm-laz}
Assume that $\PS\subset\Q[x_1,\ldots,x_n]$ has a finite positive
number of solutions over $\C$. The eigenvalues of $M_f$ are the
values of $f$ at the roots of $\PS=0$ over $\C$ with respect to
multiplicities of the roots of $\PS=0$.
\end{thm}

Let $s_i$ be rational numbers satisfying \bref{eq-dn3} and
$${\mathcal{F}_i} = \PS\cup\{x - x_1 - s_1 x_2 - \cdots - s_1\cdots
s_{i-1} x_i\}.$$

We can compute $g_i(x_i)$ and $T_i(x)$ such that
\begin{equation}\label{eq-ti}
(g_i(x_i))=\Q[x_i]\cap(\PS) \hbox{ and }
(T_i(x))=\Q[x]\cap(\mathcal{F}_i).
\end{equation}
In fact, we can construct the matrixes for $M_{x_i}$ and $M_x$ based
on $\B$ and $\GB$, and $g_i(x_i)$ and $T_i(x)$ are the minimal
polynomials for $M_{x_i}$ and $M_x$, respectively  (See reference
\cite{cox}).
Note that we can also use the method introduced in reference
\cite{fglm} to compute $g_i(x_i), T_i(x)$.

From Theorem \ref{thm-laz} and (a) of Lemma \ref{lm-ww1},  the
$i$-th coordinates of all the roots of $\PS=0$ are roots of
$g_i(x_i)=0$, and all the possible values of $x=\sum_{j=1}^i
s_1\cdots s_{j-1} x_j$ on the roots of $\PS=0$ are roots of
$T_i(x)=0$.

Now we show how to estimate $r_{i}$ defined in \bref{eq-dn2}. At
first, compute $(g_{i}(x_{i}))=(\PS)\cap\Q[x_{i}]$. Then we have the
following result.

\begin{lem}\label{lm-rb1}
Use the notations introduced before. Then
\begin{equation}\label{eq-rb2}
r_{i} = 2\max \{\RB(g_{i}(x_{i}))
\}
\end{equation}
satisfies the condition \bref{eq-dn2}, where $\RB(g)$ is the root
bound of a univariate polynomial equation $g=0$.
\end{lem}
\begin{pf}
The lemma is obvious since for any root $(\xi_1,\ldots,\xi_{i})\in
V_{\C}(\IS_{i})$, $\xi_{i}$ is a root of $g_{i}(x_{i})=0$.
\end{pf}

\subsection{Theoretical preparations for the algorithm}

In this subsection, we will outline an algorithm to compute an LUR
for $\PS$ and to isolate the roots of $\PS=0$ under a given
precision $\epsilon$. The algorithm is based on an interval version
of Theorem \ref{th-1}.

We define the {\bf isolation boxes} for an LUR defined in
\bref{eq-lur} as: {\small \begin{eqnarray}\label{eq-ibox} && \{
B_{1}\times\frac{B_{2}-B_{1}}{s_1}\times \cdots\times
\frac{B_{n}-B_{n-1}}{s_1\cdots s_{n-1}}\,|\,
 B_i\in \mathcal{B}_i, \Dis(B_{i+1},B_{i}) < \rho_i/2,1\le i\le
n-1\}
\end{eqnarray}}
where $\mathcal{B}_i$ is a set of isolation boxes for the complex
roots of $T_i(x)=0$ and $\rho_i = s_1\cdots s_{i-1}d_i$.
In  Theorem \ref{thm-eps1} to be proved below, we will give criteria
under which the
isolation boxes for $\PS$ are the isolation boxes of an LUR.\\

Let $\PS\subset\Q[x_1,\ldots,x_n]$ be a zero-dimensional polynomial
system.
We will compute an LUR for $\mathcal{P}$ and a set of
$\epsilon$-isolation boxes for the roots of $\mathcal{P}=0$
inductively.

At first, consider $i=1$. We compute $T_1(x)$ as defined in equation
\bref{eq-ti}. Let $\mathcal{B}_1$ be a set of isolation intervals
for the complex roots of $T_1(x)=0$. Then, we can set $d_1$ to be
the minimal distance between any two intervals in $\mathcal{B}_1$.

For $i$ from $1$ to $n-1$, assuming that we have computed

\quad$\bullet$ An LUR $(T_1(x),\ldots,T_i(x), s_j,d_j,
j=1,\ldots,i-1)$ for $\IS_i$.

\quad$\bullet$ A set of $\epsilon$-isolation boxes for $\IS_i$.

\quad$\bullet$ The parameter $d_i$.

We will show how to compute  $r_{i+1}$, $s_i$, $T_{i+1}(x)$,
$d_{i+1}$, and a set of $\epsilon$-isolation boxes of the roots of
$\IS_{i+1}=0$. The procedure consists of three steps.


{\bf Step 1.} We will compute $r_{i+1},s_{i}$ as introduced in
\bref{eq-dn2} and \bref{eq-dn3}. With $s_{i}$, we can compute
$T_{i+1}(x)$ as defined in \bref{eq-ti}.

Here $r_{i+1}$ can be computed with the method in Lemma
\ref{lm-rb1}. Note that $d_i$ is known from the induction
hypotheses. Then we can choose a rational number $s_i$ such that
condition \bref{eq-dn3} is valid. Finally, $T_{i+1}(x)$ can be
computed with the methods mentioned below equation \bref{eq-ti}.

{\bf Step 2.} We are going to compute the isolation intervals of the
roots of $\IS_{i+1}=0$. Let  $\xi=(\xi_1,\ldots,\xi_i)$ be a root of
$\IS_i=0$. We are going to find the roots of $\IS_{i+1}=0$ ``lifted"
from $\xi$, that is, roots of the form
\begin{eqnarray}\label{eq-n10}\zeta_j&=&(\xi_1,\ldots,\xi_i,\xi_{i+1,j}),j=1,\ldots,m.\end{eqnarray}
To do that, we need only to find a set of isolation intervals for
$\xi_{i+1,j}$ with lengths no larger than $\epsilon$, since we
already have an $\epsilon$-box for $\xi$.

Let
\begin{eqnarray}\label{eq-n11}
   \eta_i&=&\xi_1+s_1\xi_2+\cdots+s_1\cdots s_{i-1}\xi_{i}.\nonumber
\end{eqnarray}
Then, $\eta_i$ is a root of $T_{i}(x)=0$.
By (b) of Lemma \ref{lm-ww1} the roots $\theta_j$ of $T_{i+1}(x)=0$
correspond to $\zeta_j$ are
\begin{eqnarray}\label{eq-n12}
 \theta_j&=&\eta_i + s_1\cdots s_{i}\xi_{i+1,j},j=1,\ldots,m.
\end{eqnarray}
We have

\begin{lem}\label{lm-rr1}
Let $I_i=\langle[a,b],[c,d]\rangle$ be an isolation interval for the
root $\eta_i$ of $T_i(x)=0$ such that $|I_i|< \frac{1}{4}\rho_i$
where $\rho_i=s_1\cdots s_{i-1}d_{i}$. Then all $\theta_j$ in
\bref{eq-n12} are in the following complex interval
\begin{eqnarray}\label{eq-nbh1}
\mathbb{I}_{I_i} &=&\langle(a-\rho_i/2,b+\rho_i/2,
(c-\rho_i/2,d+\rho_i/2)\rangle.
\end{eqnarray}
Furthermore, the intervals $\mathbb{I}_{\eta}$ are disjoint for all
roots $\eta$ of $T_i(x)=0$.
\end{lem}
\begin{pf}
In Figure \ref{fig-box1}, let square $ABCD$ be $\mathbb{S}_{\eta_i}
= \{\theta\in \C\, |\, |\theta-\eta_i|< \rho_i \}$ and  square
$A_1B_1C_1D_1 = \{\theta\in\C\,|\, |\theta-\eta_i| < \rho_i/2\}$.
By equation \bref{eq-ww1}, we know $|\theta_{j}-\eta_{i}|<
\frac{1}{2}\rho_i.$ So, $\theta_j$ is inside $A_1B_1C_1D_1$.
Let rectangle $A_2B_2C_2D_2$ be the complex interval $I_i$ and
rectangle $A_3B_3C_3D_3$ the complex interval $\mathbb{I}_{I_i}$
which is obtained by adding $\rho_i/2$ in each direction of the
rectangle $A_3B_3C_3D_3$.
Then, $\mathbb{I}_{I_i}$ contains $A_1B_1C_1D_1$ and hence
$\theta_{j}$ is inside $\mathbb{I}_{I_i}$.

For any $\theta\in\mathbb{I}_{I_i}$, we have $|\theta-\eta_i| \le
|\theta-P|$ where $P$ is one of the points $A_2,B_2,C_2,D_2$ to make
$|\theta-P|$ maximal. It is clear that $|\theta-P|\le \rho_i/2 +
|I_i| = \frac{3}{4}\rho_i$. So, $\forall\theta\in\mathbb{I}_{I_i}$,
$|\theta-\eta_i| \le \frac{3}{4}\rho_i$. Since $\mathbb{S}_{\eta_i}$
is the set of complex numbers satisfying $|\theta-\eta_i|< \rho_i$,
we have $\mathbb{I}_{I_i}\subset\mathbb{S}_{\eta_i}$.
By (c) of Lemma \ref{lm-ww1}, $\mathbb{S}_{\eta_i}$ are disjoint for
all roots of $T_i(x)=0$. Then  $\mathbb{I}_{I_i}$ are disjoint for
all roots of $T_i(x)=0$ too.
\end{pf}

\begin{figure}[ht]
\centering
\begin{minipage}{0.7\textwidth}
\centering
 \includegraphics[scale=0.40]{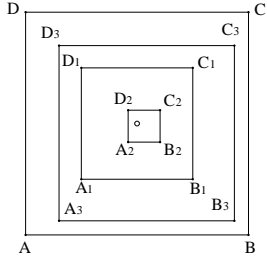}
 \caption{The isolation intervals $I_i$, $\mathbb{S}_{\eta_i}$, $\mathbb{I}_{I_i}$ for a root $\eta_i$ of $T_i(x)=0$.
 $\eta_i$ is represented by $\circ$.}
\label{fig-box1}
\end{minipage}
\end{figure}

The following lemma shows what is the precision needed to isolate
the roots of $T_{i+1}(x)=0$ in order for the isolation boxes to be
contained in some $\mathbb{I}_{I_i}$. It can be seen as an effective
version of the fact $\eta_{i+1}\in \mathbb{S}_{\eta_i}$ proved in
Lemma \ref{lm-ww0}.
\begin{lem}\label{lm-rr2}
Use the notations introduced in Lemma \ref{lm-rr1}. Let $\{ B_j,
j=1,\ldots,m \}$ be a set of $\frac{1}{4}\rho_i$-isolation boxes for
the roots $\theta_j, j=1,\ldots,m$ of $T_{i+1}(x)=0$. Then, for all
$j$
\begin{eqnarray}\label{eq-rr2}
B_j\subset\mathbb{I}_{I_i} \hbox{ and } \Dis(B_j,I_i) < \rho_i/2.
\end{eqnarray}
\end{lem}
\begin{pf}
From the proof of Lemma \ref{lm-rr1},  the distance from $\eta_i$ to
line $BC$ in Figure \ref{fig-box1} is $\rho_i$ and the distance from
$\eta_i$ to line $B_3C_3$ is less than $\frac{3}{4}\rho_i$. So, the
distance between line $BC$ and $B_3C_3$ is at least
$\frac{1}{4}\rho_i$. This fact is also valid for the pairs of lines
$AD/A_3D_3$, $AB/A_3B_3$, and $CD/C_3D_3$. Since the isolation boxes
$B_j$ are of size smaller than $\rho_i/4$ and their centers are
inside $A_3B_3C_3D_3$, the boxes $B_j$ must be inside $ABCD$.
Note that $I_i$ is rectangle $A_2B_2C_2D_2$. Since $\theta_j$ is
inside both $B_j$ and rectangle $A_3B_3C_3D_3$ and the distance from
$\eta_i$ to each edge of $A_3B_3C_3D_3$ is $\rho_i/2$, the distance
between  $B_j$ and $I_i$ must be smaller than $\rho_i/2$.
\end{pf}

Isolate the roots of $T_{i+1}(x)=0$ with precision
$\frac{1}{4}\rho_i$. By Lemma \ref{lm-rr2}, all the roots are in one
of the intervals $\mathbb{I}_{I}$ where $I$ is an isolation interval
for a root $\eta$ of $T_i(x)=0$.

Let $K_j=\langle[p_j,q_j],[g_j,h_j]\rangle(1\le j\le m)$ be the
isolation intervals for the roots $\theta_j$ of $T_{i+1}(x)=0$
inside $\mathbb{I}_{I_i}$.
Then from \bref{eq-n12}, the isolation intervals of
$\xi_{i+1,j}(1\le j\le m)$ are
\begin{eqnarray}\label{eq-inti}J_{i+1,j}=\frac{K_j - I_i}{s_1\cdots s_i}=\frac{\langle[p_j-b,q_j-a],[g_j-d,h_j-c]\rangle}{s_1\cdots s_i}. \end{eqnarray}

We have
\begin{lem}With the notations introduced above, if the following
conditions
\begin{eqnarray}\label{eq-p1} (q_j-p_j)+(b-a)<s_1\cdots s_i\epsilon,\,\,\, (h_j-g_j)+(d-c)<s_1\cdots s_i\epsilon\end{eqnarray}
\begin{equation}\label{eq-cd}
S_{\eta_i}=\min\limits_{1\le k\neq j\le
m}\Dis(\langle[p_k,q_k],[g_k,h_k]\rangle,\langle[p_j,q_j],[g_j,h_j]\rangle)>\max\{b-a,d-c\}.
\end{equation}
are valid, then $J_{i+1,j}$ defined in \bref{eq-inti} are
$\epsilon$-isolation intervals of $\xi_{i+1,j}$ in equation \bref
{eq-n10}.
\end{lem}
\begin{pf}
%
It is clear that condition \bref{eq-p1} is used to ensure the
precision: $|J_{i+1,j}| < \epsilon$.

We consider \bref{eq-cd} below. Assume that $J_{i+1,j}, J_{i+1,k}
(1\le k\neq j\le m)$ are any two intervals defined in \bref{eq-inti}
for the $(i+1)$-th coordinates of the roots
$(\xi_1,\ldots,\xi_i,\xi_{i+1,j})$,
$(\xi_1,\ldots,\xi_i,\xi_{i+1,k})$ of $\IS_{i+1}=0$, respectively.
Since we have derived the $\epsilon$-isolation boxes for the roots
of $\IS_i=0$, we need only to ensure that the intervals of the
$(i+1)$-th coordinates of the roots of $\IS_{i+1}=0$ lifted from a
fixed root of $\IS_{i}=0$ are isolation intervals. That is, to show
$\Dis(J_{i+1,j},J_{i+1,k})>0$.

Assume that $K_j=\langle[p_j,q_j],[g_j,h_j]\rangle$ and
$K_k=\langle[p_k,q_k],[g_k,h_k]\rangle$ are the isolation intervals
of the roots $\eta_j$, $\eta_k$ of $T_{i+1}(x)=0$. Here $\eta_j$,
$\eta_k$ correspond to $(\xi_1,\ldots,\xi_i,\xi_{i+1,j})$,
$(\xi_1,\ldots,\xi_i,\xi_{i+1,k})$, respectively. So $K_j,K_k$
correspond to $J_{i+1,j},J_{i+1,k}$, respectively. Assume that
$p_j\le p_k,g_j\le g_k$. Then we have {\small
$$\Dis(J_{i+1,j},J_{i+1,k})=\frac{\max\{\Dis([p_j-b,q_j-a],[p_k-b,q_k-a]),\Dis([g_j-d,h_j-c],[g_k-d,h_k-c])\}}{s_1\cdots
s_i},$$} and
$$\mathcal{L}_1=\Dis([p_j-b,q_j-a],[p_k-b,q_k-a])=
\begin{cases}
(p_k-q_j)-(b-a), \,\hbox{ if } (p_k-q_j)-(b-a)>0,\\
0, \hspace{1cm}\hbox{        otherwise},
\end{cases}
$$
$$\mathcal{L}_2=\Dis([g_j-d,h_j-c],[g_k-d,h_k-c])=
\begin{cases}
(g_k-h_j)-(d-c), \,\hbox{ if } (g_k-h_j)-(d-c)>0,\\
0, \hspace{1cm}\hbox{        otherwise}.
\end{cases}
$$
$K_j$ and $K_k$ are disjoint since they are isolation intervals of
$T_{i+1}(x)=0$. So
$$\Dis(K_j,K_k)=\max\{p_k-q_j,g_k-h_j\}>0.$$
It is clear that $\Dis(J_{i+1,j},J_{i+1,k})>0$ if $\mathcal{L}_1>0$ or $\mathcal{L}_2>0$.
Then we conclude if inequality \bref{eq-cd} is true, then
$\Dis(J_{i+1,j},J_{i+1,k})>0$. This proves the lemma.
\end{pf}

Geometrically, $S_{\eta_i}$ is the separation bound for the roots of
$T_{i+1}(x)=0$ corresponds to those roots of $\IS_{i+1}$ lifted from
the root of $\IS_i=0$ corresponding to the root $\eta_i$ of
$T_i(x)=0$.

\begin{rem}\label{rem-2} Note that in \bref{eq-cd}, we obtain
$I_i=\langle[a,b],[c,d]\rangle$ first and
$K_j=\langle[p_j,q_j],[g_j,h_j]\rangle$ later. We will refine the
isolation interval $I_i$ of $T_i(x)=0$ such that inequality
\bref{eq-cd} is true. After the refinement, all other conditions are
still valid. We need to do this kind of refinement only once.
\end{rem}

As a consequence of the above lemma, we have
\begin{cor}\label{cor-b2}
Let $\mathbb{B}$ be an $\epsilon$-isolation box for the root $\xi$
of $\IS_{i}=0$ and $J_{i+1,j}$ defined in \bref{eq-inti}. If
\bref{eq-p1}, \bref{eq-cd} are valid, then $\mathbb{B}\times
J_{i+1,j},j=1,\ldots,m$ are $\epsilon$-isolation boxes for the roots
$\rho_j$ of $\IS_{i+1}=0$, which are lifted from $\xi$.
\end{cor}

{\bf Step 3.} We will show how to compute $d_{i+1}$ from the
isolation intervals of $T_{i+1}(x)=0$.

\begin{lem}\label{lem-di}
Let
\begin{equation}\label{eq-di}
d_{i+1}=\min\{\frac{S_{i+1}}{2s_1\cdots s_i}, \frac{d_i}{2s_i}\},
\end{equation}
where $S_{i+1}$ is the minimal distance between any two isolation
intervals of $T_{i+1}(x)=0$. Then $d_{i+1}$ satisfies conditions
\bref{eq-dn1}.
\end{lem}
\begin{pf}
Let $\theta_j$ and $\theta_k$ be two different roots of
$T_{i+1}(x)=0$ defined in \bref{eq-n12}. Then we have
$$\xi_{i+1,j}-\xi_{i+1,k}=\frac{\theta_j-\theta_k}{s_1\ldots s_{i}}.$$
Therefore,  $D_{i+1}= \min_{\eta\in V_{\C}(T_{i}(x))}\{
\frac{S_{\eta}}{2s_1\cdots s_i}\}$ is the parameter defined in
\bref{eq-dn0}, where $S_{\eta}$ is determined as in \bref{eq-cd}. It
is clear that $D_{i+1}$ is not larger than $S_{i+1}$ which is the
minimal distance between any two isolation intervals of
$T_{i+1}(x)=0$. Then, the first condition in \bref{eq-dn1} is
satisfied.
In order for the second condition in \bref{eq-dn1} to be satisfied,
we also require $d_{i+1}\le \frac{d_i}{2 s_i}.$ So the lemma is
proved.
\end{pf}



We can summarize the result as the following theorem which is an
interval version of Theorem \ref{th-1}.

\begin{thm}\label{thm-eps}
Let \bref{eq-lur} be an LUR such that $d_i$, $r_i$, and $s_i$
satisfy \bref{eq-di}, \bref{eq-dn2}, and \bref{eq-dn3} respectively,
$\mathcal{B}_{i}$ the $\epsilon_i$-isolation boxes for the roots of
$T_i(x)=0$, and $S_{i} = \min\{\Dis(B_1, B_2)\,|\,
B_1,B_2\in\mathcal{B}_i, B_1\ne B_2 \}$. If
\begin{eqnarray}\label{eq-eps1}
\epsilon_1\le\epsilon, \epsilon_i+\epsilon_{i+1}\le s_1\cdots s_i
\epsilon,\,\,
 \epsilon_i\le \frac{\rho_i}{4},\,
 \epsilon_{i+1}\le \frac{\rho_i}{4},\,
 \epsilon_i\le S_{i+1}, \end{eqnarray} where $\rho_i=s_1\cdots s_{i-1}\,d_{i}$,
 then \bref{eq-ibox} is a set of $\epsilon$-isolation boxes  for $\PS$.
\end{thm}
\begin{pf}
%
We  first explain the functions of the inequalities in
\bref{eq-eps1}.
The first two inequalities in \bref{eq-eps1} are introduced in
\bref{eq-p1} to ensure the $\epsilon$ precision for the isolation
boxes.
The third inequality in \bref{eq-eps1} is introduced in Lemma
\ref{lm-rr1} to ensure $\theta_j\in\mathbb{I}_{I_i}$ and
$\mathbb{I}_{I_i}$ are disjoint.
The fourth inequality is introduced in Lemma \ref{lm-rr2} to ensure
the isolation intervals of the roots of $T_{i+1}(x)=0$ are inside
their corresponding interval $\mathbb{I}_{I_i}$.
The last inequality is introduced in \bref{eq-cd} to ensure the
recovered isolation boxes of $\mathcal{P}$ are disjoint.

Now the theorem is a consequence of Corollary \ref{cor-b2}. Here, we
give the explicit expression for the isolation boxes. The expression
for interval $J_{i+1,j}$ in \bref{eq-inti} is directly given. The
matching condition $\Dis(B_{i+1},B_{i}) < \rho_i/2$ is from
condition \bref{eq-rr2}.
\end{pf}

We have the following effective version of Theorems \ref{th-1} and
\ref{thm-eps} by giving an explicit formula for $\epsilon_i$.
\begin{thm} \label{thm-eps1} Using the same notations as Theorem \ref{thm-eps}.
Let $\epsilon$ be the given precision to isolate the roots of
$\mathcal{P}$.
Let
\begin{eqnarray}\label{eq-eps3}
 \epsilon_1 &=& \min\{\epsilon,\frac{s_1\epsilon}{2},\frac{d_{1}}{4},
 S_{2}\},\nonumber\\
\epsilon_i&=&\min\{\frac{s_1\cdots
s_{i-1}\epsilon}{2},\frac{s_1\cdots s_i\epsilon}{2},
     \frac{s_1\cdots s_{i-1}d_{i}}{4}, \frac{s_1\cdots s_{i-2}d_{i-1}}{4}, S_{i+1}\},
\end{eqnarray}\
where $i=2,...,n$, $s_0=1,s_n=1$, $S_{n+1}=+\infty$.
If we isolate the roots of $T_i(x)=0$ with precision $\epsilon_i$,
then  \bref{eq-ibox} is a set of $\epsilon$-isolation boxes  for
$\PS=0$.
\end{thm}
\begin{proof}
By \bref{eq-eps3}, we have
 $\epsilon_i\le\frac{s_1\cdots s_{i}\epsilon}{2}$ and
 $\epsilon_{i+1}\le\frac{s_1\cdots s_{i}\epsilon}{2}$.
Then the second inequality in \bref{eq-eps1},
$\epsilon_i+\epsilon_{i+1}\le s_1\cdots s_i \epsilon$, is valid. All
other inequalities in \bref{eq-eps1} are clearly satisfied and the
theorem is proved.
\end{proof}

We can also compute the multiplicities of the roots with the LUR
algorithm.

\begin{cor}
If we compute the last univariate polynomial $T_n(x)$ in the LUR as
the characteristic polynomial of $M_x$, then the multiplicities of
the roots of $\PS=0$ are the multiplicities of the corresponding
roots of $T_n(x)=0$.
\end{cor}
\begin{proof}
By (a) of Lemma \ref{lm-ww1}, $x = x_1+s_1\,x_2+\cdots+s_1\cdots
s_{n-1}x_n$ is a separating element. By Theorem \ref{thm-laz}, the
characteristic polynomial of $M_x$ keeps the multiplicities of the
roots of the system. The corollary is proved.
\end{proof}

\subsection{Algorithm}

Now, we can give the full algorithm based on LUR.

\begin{alg} \label{alg-1}
The input is a zero dimensional polynomial system
$\PS=\{P_1,\dots,P_t\}$ in $\Q[x_1,\ldots,x_n]$ and a positive
rational number $\epsilon$. The output is an LUR for $\PS$ and a set
of $\epsilon$-isolation boxes for the roots of $\PS=0$.
\end{alg}

\begin{description}
\item[S1]
Compute a Gr$\ddot{o}$bner basis $\GB$ of $\PS$ with the graded
reverse lexicographic order and a monomial basis $\B$ for linear
space $\mathcal{A}=\Q[x_1,\ldots,x_n]/(\PS)$ over $\Q$.

\item[S2]
Compute $T_1(x)$ as defined in \bref{eq-ti} with the method given in
Section 3.2; compute a set of $\epsilon$-isolation boxes
$\mathcal{B}_1$ for the complex roots of $T_1(x)=0$; set $d_1 =
\min\{\Dis(B_1,B_2)\, |\, B_1,$ $B_2\in\mathcal{B}_1, B_1\ne B_2,
\}$.

\item[S3]
For $i=1,\ldots,n-1$, do steps  {\bf S4}-{\bf S9}; output the set of
boxes \bref{eq-ibox}.

\item[S4]
Compute $r_{i+1}$ with the method in Lemma
\ref{lm-rb1}. Select a rational number $s_i$ such that condition
\bref{eq-dn3} is valid.

\item[S5]
Compute $T_{i+1}(x)$ as defined in \bref{eq-ti} with the method
given in Section 3.2.

\item[S6]
Set $\rho_i=s_1\cdots s_{i-1}d_{i}$ and compute a set of
$\frac{1}{4}\rho_i$-isolation boxes $\mathcal{B}_{i+1}$ for the
complex roots of $T_{i+1}(x)=0$

\item[S7]
Set $S_{i+1} = \min\{\Dis(B_1,B_2)\,|\,B_1,B_2\in\mathcal{B}_{i+1},
B_1\ne B_2\}$.

\item[S8]
Compute $d_{i+1}$ with formula \bref{eq-di}.

\item[S9]
Compute $\epsilon_i$ with formula \bref{eq-eps3}; refine the
isolation boxes $\mathcal{B}_i$ of $T_i(x)=0$ with precision
$\epsilon_i$; still denote the isolation boxes as $\mathcal{B}_i$.

\end{description}

\begin{rem}
From Lemma \ref{lm-rr1}, the roots of $T_{i+1}(x)=0$ are in the
rectangle $\mathbb{I}_{I_i}$.
So, we need only to isolate the roots of $T_i(x)=0$ inside these
rectangles. This property is very useful in practice, see Figure
\ref{fig-slope} for an illustration.
\end{rem}

\section{Examples}
In this section, we will give some examples to illustrate our
method.

We first use the following  example to show how to isolate the roots
of a system with our method. Note that with an LUR, we can also use
floating point numbers to compute the roots of $\PS=0$ if the
floating point number can provide the required precision as shown in
the following example.

\begin{exmp} Consider the system
$\mathcal{P}:=[x^2+y^2+z^2-3,x^2+2y^2-3z+1,x+y-z].$ The coordinate order is $(x,y,z)$.

The Gr$\ddot{o}$bner basis $\GB$ with the graded reverse
lexicographic order $z>y>x$ of $\PS$ is: {\small
$$[-x-y+z,{x}^{2}+2\,yx+3\,x-4+3\,y,-3\,x+{x}^{2}+1-3\,y+2\,{y}^{2},6\,{
x}^{3}-30+9\,{x}^{2}+25\,y+5\,x].$$} The leading monomials of the
basis are $\{z,x\,y,y^2,x^3\}$. So the monomial basis of the
quotient algebra $\mathcal{A}=\Q[x_1,...,x_n]/(\PS)$ is $\B=[1, x,
y, x^2]$.

Let $t=x$, we can compute:
$$M_t=\left[ \begin {array}{cccc} 0&1&0&0\\ \noalign{\medskip}0&0&0&1\\ \noalign{\medskip}2&-3/2&-3/2&-1/2\\ \noalign{\medskip}5&-5/6&-{
\frac {25}{6}}&-3/2\end {array} \right].$$ The minimal polynomial of
$M_t$ is $$T_1(t)=5-60\,t+6\,t^2+18\,t^3+6\,t^4.$$ Compute its
complex roots with the function ``Analytic'' in Maple package
[RootFinding], we obtain
\begin{eqnarray*}
 R_1&=&[- 2.22081423399575-1.53519779646152\,\mathfrak{i}, - 2.22081423399575\\
 &&+ 1.53519779646152\,\mathfrak{i},
  0.0842270424726020, 1.35740142551890].
\end{eqnarray*}
Computing the roots distance with formula \bref{eq-rtsep}, we obtain
$d_1\le0.6365871918$. We can set $$d_1=\frac{1}{2}.$$

In a similar way, we compute $M_y$ and its minimal polynomial
$g_2(y)=-29-66\,y+60\,y^2+12\,y^4$. The root bound of $g_2(y)$ is
$3$. So we have $r_2=6$. Since $\frac{d_1}{r_2}=\frac{1}{12}$, we
set $$s_1=\frac{1}{20}.$$ Let $t=x+s_1\,y$. We can compute a matrix
$M_t$ and its minimal polynomial
$$T_2(t)=863337-6119640\,t+360000\,t^2+1920000\,t^3+640000\,t^4.$$
Computing its complex roots, we have
\begin{eqnarray*}
 R_2&=&[-2.24194942371773 -1.41342395552762\mathfrak{i}, -2.24194942371773\\
 &&+1.41342395552762 \mathfrak{i}, 0.143249906267126, 1.34064894116850].
\end{eqnarray*}
Computing the minimal distance between any two roots, we have
$S_2=0.5986995174$. From equation \bref{eq-di}, we can obtain
$$d_2=\min\{\frac{S_2}{2\,s_1},\frac{d_1}{2\,s_1}\}=5.$$

Compute $M_z$ and its minimal polynomial $g_3(z)=121-132 z-36 z^2+36
z^3+12 z^4$. Then the root bound of $g_3(z)$ is $5$. We have
$r_3=10$. We can set $$s_2=\frac{1}{2}\le
\frac{d_2}{r_3}=\frac{1}{2}.$$ Let $t=x+s_1\,y+s_1 s_2 z$. Compute
$M_t$ and its minimal polynomial
$$T_3(t)=53294617-309903360\,t\\+11884800\,t^2+94464000\,t^3+30720000\,t^4.$$
Computing its complex roots, we have
\begin{eqnarray*}
 R_3&=&[- 2.30803737442857- 1.39091697997219\,\mathfrak{i},- 2.30803737442857\\
 &&+1.39091697997219\,\mathfrak{i}, 0.174867014226204, 1.36620773463121].
\end{eqnarray*}

We use $R_1[i]$ to represent the $i$-th element of $R_1$. $R_2[i],
R_3[i]$ are similarly defined. Since
$R_2[1]-R_1[1]=-0.021135190+0.121773840 \mathfrak{i}$ and the
absolute values of its  real part and imaginary part are lese than
$1/2$, $(R_1[1],\frac{R_2[1]-R_1[1]}{s_1})$ is a root of
$\PS\cap\Q[x,y]$. But $R_2[2]-R_1[1]=-0.021135190+2.948621752
\mathfrak{i}$ and its imaginary part is larger than $1/2$. Then
$R_2[2]$ does not correspond to $R_1[1]$.
$R_3[1]-R_2[1]=-0.066087950+0.022506976 \mathfrak{i}$ and the
absolute values of its real part and imaginary part are lese than
$1/4$, so
\begin{eqnarray*}
&&(R_1[1], \frac{R_2[1]-R_1[1]}{s_1}, \frac{R_3[1]-R_2[1]}{s_1 s_2})\\
&=&(-2.22081423399575-1.53519779646152\, \mathfrak{i}, -0.42270380+2.43547680\, \mathfrak{i},\\
 &&-2.64351800+0.90027904\, \mathfrak{i})
\end{eqnarray*}
is a root of $\PS=0$. In a similar way, we can find all other
complex roots of $\PS=0$. And from Theorem \ref{thm-eps1}, we can
set $\epsilon_1=\frac{1}{40}\epsilon,
\epsilon_2=\epsilon_3=\frac{1}{80}\epsilon$, where $\epsilon$ is the
given precision. So if we refine the roots of $T_i(t)=0$ to five
digits, we can obtain the roots of $\PS=0$ with three digits.

We also obtain an LUR for $\mathcal{P}$ as follows:
$$[[T_1(t),T_2(t),T_3(t)],[s_1,s_2],[d_1,d_2]].$$

The roots of $\mathcal{P}=0$ are: $$[(
\alpha,20(\beta-\alpha),40(\gamma-\beta) )
|T_1(\alpha)=0,T_2(\beta)=0,T_3(\gamma)=0,|\beta-\alpha|<1/2,|\gamma-\beta|<1/4].$$

Assuming that the final precision for the real roots of the system
is $\epsilon=1/2^{10}$ and isolating the real roots of $T_i(t)=0$
with precision $\epsilon_1=\frac{1}{40}\epsilon,
\epsilon_2=\epsilon_3=\frac{1}{80}\epsilon$, respectively, we can
obtain the following two real roots of $\PS=0$ with the given
precision: {\footnotesize
$$[\frac{5519}{65536}, \frac{345}{4096}]\times[\frac{4835}{4096}, \frac{38695}{32768}]\times[\frac{20715}{16384}, \frac{20725}{16384}],\,\,[\frac{44479}{32768}, \frac{88959}{65536}]\times[\frac{-10985}{32768}, \frac{-5485}{16384}]\times[\frac{16745}{16384}, \frac{16755}{16384}].$$
}
\end{exmp}
In the next example, we will compare our method with RUR \cite{rur}.
\begin{exmp} Consider the following example from paper \cite{rur}.
$\mathcal{P}:=[24\,uz-{u}^{2}-{z}^{2}-{u}^{2}{z}^{2}-13,24\,yz-{y}^{2}-{z}^{2}-{y}^{
2}{z}^{2}-13,24\,uy-{u}^{2}-{y}^{2}-{u}^{2}{y}^{2}-13].$ The
coordinate order is $(u,y,z)$.

The RUR is as follows.
$$f(x)=0,\,u=\frac{g(u,x)}{g(1,x)},\,y=\frac{g(y,x)}{g(1,x)},\,z=\frac{g(z,x)}{g(1,x)},$$
where
\begin{eqnarray}
f(x)&=&{x}^{16}-5656\,{x}^{14}+12508972\,{x}^{12}-
14213402440\,{x}^{10}
+9020869309270\,{x}^{8}\nonumber\\&&-3216081009505000\,{x}^{6}
+606833014754230732\,{x}^{4}\nonumber\\ &&-51316296630855044152\,{x}^{2}+1068130551224672624689,\nonumber\\
g(1,x)&=&{x}^{15}-4949
\,{x}^{13}+9381729\,{x}^{11}-8883376525\,{x}^{9}+4510434654635\,{x}^{7
}\nonumber\\&&-1206030378564375\,{x}^{5}+151708253688557683\,{x}^{3}-
6414537078856880519\,x,\nonumber\\
g(u,x)&=&116\,{x}^{14}-483592\,{x}^{12}+
784226868\,{x}^{10}-634062241592\,{x}^{8}\nonumber\\ &&+270086313707548\,{x}^{6}-
58355579408017944\,{x}^{4}+5520988105236180668\,{x}^{2}\nonumber\\ &&-131448117382500870952,\nonumber\\
g(y,x)&=&86\,{x}^{14}
-418870\,{x}^{12}+759804846\,{x}^{10}-670485664238\,{x}^{8}+
307445009725282\,{x}^{6}\nonumber\\ &&-71012402366579778\,{x}^{4}+
7099657810552674458\,{x}^{2}-168190996202566563226,\nonumber\\
g(z,x)&=&71\,{x}^{14}-355135\,{x}^{12}+673508751\,{x}^{10
}-633214359791\,{x}^{8}+314815356659869\,{x}^{6}\nonumber\\ &&-79677638700441717\,{x
}^{4}+8618491509948092045\,{x}^{2}-205956089289536014429.\nonumber
\end{eqnarray}

 An LUR of $\mathcal{P}$ is as follows:
{\small $$[[T_1(t),T_2(t),T_3(t)],[s_1,s_2],[d_1,d_2]]=[[T_1(t),T_2(t),T_3(t)], [1/200, 1/15], [0.2237374734, 2.146554200]],$$}
where
{\small
\begin{eqnarray}
T_1(t)&=&169-1820\,{t}^{2}+2622\,{t}^{4}-140\,{t}^{6}+{t}^{8},\nonumber\\
T_2(t)&=&12034552627604020308981441166197-133523438810776274535699687120000\,{t
}^{2}\nonumber\\&&+334257305564156882138712000000000\,{t}^{4}-
256456971612085383936000000000000\,{t}^{6}\nonumber\\&&+
23629005541670400000000000000000\,{t}^{8}-
665288908800000000000000000000\,{t}^{10}\nonumber\\&&+4096000000000000000000000000
\,{t}^{12},\nonumber\\
T_3(t)&=&398658124842757922827990174525891734024598098970801\nonumber\\&&-
5057045016775809265742737650285696238919118781687500\,{t}^{2}\nonumber\\&&+
18306568462902747682078658662680830721818866699218750\,{t}^{4}\nonumber\\&&-
26971016274307991838575084944533427932357788085937500\,{t}^{6}\nonumber\\&&+
15563591910271113423505114668403939783573150634765625\,{t}^{8}\nonumber\\&&-
1936419155067693199961145026385784149169921875000000\,{t}^{10}\nonumber\\&&+
94190634217706926258139312267303466796875000000000\,{t}^{12}\nonumber\\&&-
1851048158439662307500839233398437500000000000000\,{t}^{14}\nonumber\\&&+
10022595757618546485900878906250000000000000000\,{t}^{16}.\nonumber
\end{eqnarray}
}

The roots of $\mathcal{P}$ are: $\{(u,y,z)=(
\alpha,200(\beta-\alpha),3000(\gamma-\beta) )
|T_1(\alpha)=0,T_2(\beta)=0,T_3(\gamma)=0,|\beta-\alpha|<0.2237374734,
|\gamma-\beta|<0.01073277100\}.$

\end{exmp}

\section{Conclusion}
We give a new representation, LUR, for the roots of a
zero-dimensional polynomial system $\mathcal{P}$ and propose an
algorithm to isolate the roots of $\mathcal{P}$ under a given
precision $\epsilon$. For the LUR, the roots of the system are
represented as a linear combination of the roots of some univariate
polynomial equations. The main advantage of LUR is that precision
control of the roots of the system is much clearer.

The main drawback of the LUR method is that when the parameters
$s_i$ becomes very small, the coefficients of $T_i(t)$ could become
very large, which will slow down the algorithm. To improve the
efficiency of the LUR algorithm is our future work.


\begin{thebibliography}{99}
\bibitem{abrw}
M.~E. Alonso, E.~Becker, M.~F. Roy, and T.~W{\"{o}}rmann.
\newblock Zeros, multiplicities, and idempotents for zerodimensional systems.
\newblock In {\em Algorithms in Algebraic Geometry and Applicatiobns},
  1--15. Birkhauser, 1996.

%
%

\bibitem{canny}
J.~F. Canny.
\newblock Some algebraic and geometric computation in pspace.
\newblock In {\em ACM Symp. on Theory of Computing}, 460--469. SIGACT,
  1988.
%
\bibitem{lgp-bi}
J.~S. Cheng, X.~S. Gao, J. Li,
\newblock Root isolation for bivariate polynomial
systems with local generic position method.
\newblock {\em Proc. ISSAC 2009},
103-109, ACM Press, 2009.
%
\bibitem{chengjsc}
J.~S. Cheng, X.~S. Gao, and C.~K. Yap.
\newblock Complete numerical isolation of real roots in zero-dimensional
  triangular systems.
\newblock {\em Journal of Symbolic Computation}, 44(7): 768--785, 2009.

\bibitem{collins-c1}
G.~E. Collins and W. Krandick.
\newblock A tangent-secant method for
polynomial complex root calculation.
\newblock {\em Proc. ISSAC 1996},
137-141, ACM Press, 1996.


\bibitem{cox}
D.~A. Cox.
\newblock Solving equations via algebras.
\newblock In {\em Solving Polnomial Equations}, Editors: Alicia Dichenstein $\&$ Ioannis Z. Emiris, Springer, 2005.

\bibitem{fglm} J. C. Faug\`{e}re, P. Gianni, d. Lazard, and T. Mora,
\newblock Efficient computation of zero-dimensional Gr\"{o}bner basis by
changing of order.
\newblock {\em Journal of Symbolic Computation}, 16(4): 329-344, 1993.
%

\bibitem{gaochou}
X.~S. Gao and S.~C. Chou.
\newblock On the theory of resolvents and its applications.
\newblock {\em Sys. Sci. and Math. Sci.}, 12: 17--30, 1999.


\bibitem{gm}
P. Gianni and T. Mora.
 \newblock Algebraic solution of systems of polynomial equations using Groebner bases.
 \newblock {\em AAECC5}, LNCS 356, 247-257, 1989.



\bibitem{gh}
M.~Giusti and J.~Heintz.
\newblock Algorithmes - disons rapides -pour la d\`ecomposition d'une
  vari\`et\'e alg\'ebrique en composantes irr\'educibles et
  \'equidimensionnelles.
\newblock In {\em Proc MEGA' 90}, pages 169--193. Birkh{\"a}user, 1991.
%

\bibitem{gls} M. Giusti, G. Lecerf, and B. Salvy,
\newblock em A Gr$\ddot{o}$bner free alternative for polynomial system solving.
 \newblock {\em Journal of Complexity}, 17: 154-211, 2001.
%

\bibitem{kmh}
H.~Kobayashi, S.~Moritsugu, and R.~W. Hogan.
\newblock Solving systems of algebraic equations.
\newblock {\em Proc. ISSAC 1988}, 139--149, ACM Press, 1988.

\bibitem{kff}
H. Kobayashi, T. Fujise, and A.~ Furukawa.
\newblock Solving systems of algebraic equations by a general elimination method.
\newblock {\em Journal of Symbolic Computation}, 5(3): 303--320, 1988.

\bibitem{kronecker}
L. Kronecker.
\newblock Grundz$\ddot{u}$ge einer arithmetischen theorie der algebraischen gr$\ddot{o}$ssen.
\newblock {\em J. Reine Angew. Math.} 92: 1-22,1882.

\bibitem{ll}
Y.N. Lakshman and D. Lazard,
\newblock On the complexity of zero-dimensional algebraic systems.
\newblock In ``Effecitve Methods in Algebraic Geometry," Progess in Mathematics, 94: 217-225, Birkh$\ddot{a}$user,Basel, 1991.


\bibitem{lazard1}
D. Lazard.
\newblock Resolution des Systemes d'Equations Algebriques.
\newblock {\em Theoretical Computer Science}, 15: 77-110, 1981.

\bibitem{intervalbook}
A. Neumaier.
\newblock Interval methods for systems of equations.
\newblock Cambridge University Press, 1990.

\bibitem{pinkert}
J. R. Pinkert.
\newblock An exact method for finding the roots of a complex
polynomial.
\newblock {\em ACM Transactions on Mathematical Software} 2(4): 351-363, 1976.  


\bibitem{renegar}
J. Renegar,
\newblock On the computaional complexity and geometry of the first-order theoery of the reals.
\newblock Part I, {\em Journal of Symbolic Computation}, 13: 255-299, 1992.

\bibitem{rur}
F.~Rouillier.
\newblock Solving zero-dimensional systems through the rational univariate
  representation.
\newblock {\em Applicable Algebra in Engineering, Communication and Computing},
  9(5): 433--461, 1999.
%
\bibitem{RZ03} F. Rouillier and P. Zimmermann.
\newblock Efficient isolation of
polynomial real roots.
\newblock {\em J. of Comp. and App. Math.}, {162}(1):
33-50, 2003.

\bibitem{sy}
M. Sagraloff and C. K. Yap.
\newblock An efficient exact subdivision algorithm for isolating complex roots
of a polynomial and its complexity analysis.
\newblock Submitted, Oct. 2009.

\bibitem{wilf}
H. S. Wilf.
\newblock A global bisection algorithm for computing the
zeros of polynomials in the complex plane.
\newblock {\em Journal of the ACM}, 25(3): 415-420, 1978. 
%

\bibitem{ynt}
K.~Yokoyama, M.~Noro, and T.~Takeshima.
\newblock Computing primitive elements of extension fields.
\newblock {\em Journal of Symbolic Computation}, 8(6): 553--580, 1989.
\end{thebibliography}

\end{document}